\documentclass{sn-jnlmod}% Math and Physical Sciences Reference Style
\usepackage{amsmath, amssymb, amsthm}
\theoremstyle{thmstyleone}%
\newtheorem{theorem}{Theorem}%  meant for continuous numbers
\newtheorem{proposition}[theorem]{Proposition}%

\theoremstyle{thmstyletwo}%
\newtheorem{remark}{Remark}%

\theoremstyle{thmstylethree}%
\newtheorem{definition}{Definition}%
\newtheorem{lemma}{Lemma}
\newtheorem{corollary}{Corollary}
\renewcommand{\p}{\partial}

\newcounter{mnotecount}[section]

\newcommand{\mnotex}[1]%{}
{\protect{\stepcounter{mnotecount}}$^{\mbox{\footnotesize $\bullet$\themnotecount}}$
\marginpar{%\color{red}%
\raggedright\tiny\em
$\!\!\!\!\!\!\,\bullet$\themnotecount: #1} }

\begin{document}


\title{Results on Lorentzian metric spaces}


\author{\fnm{} \sur{E. Minguzzi}}
%\footnote{Corresponding author.} \, }

\affil{\small \orgdiv{Dipartimento di Matematica}, \orgname{Universit\`a
	degli Studi di Pisa}, \orgaddress{\street{Largo
B. Pontecorvo 5,  I-56127}, \city{Pisa},  \country{Italy},
%e-mail:
ettore.minguzzi@unipi.it}}
%ORCID: 0000-0002-8293-3802


\abstract{
We provide a short introduction to ``Lorentzian metric spaces" i.e., spacetimes defined solely in terms of  the two-point Lorentzian distance. As noted in previous work, this structure is essentially unique if minimal conditions are imposed, such as the continuity of the Lorentzian distance and the relative compactness of chronological diamonds. The latter condition is natural for interpreting these spaces as low-regularity versions of globally hyperbolic spacetimes. Confirming this interpretation, we prove that every Lorentzian metric space admits a Cauchy time function. The proof is  constructive for this general setting and it provides a novel argument that is interesting already for smooth spacetimes.

% a result that is also of independent interest for smooth spacetimes.
%We provide a short introduction to  ``Lorentzian metric spaces'' i.e., versions of spacetime based solely on the two-point Lorentzian distance. We emphasize that, as noted in previous works, this structure is essentially unique if minimal conditions of continuity and relative compactness of diamonds are imposed - the relative compactness of chronological diamonds being natural for justifying the interpretation of these spaces as low regularity versions of globally hyperbolic spacetimes. In this work, in confirmation of this interpretation, we prove that every Lorentzian metric space admits a Cauchy time function. The proof is interesting also for smooth spacetimes.

%We provide an introduction to  ``Lorentzian metric spaces'' which are  metric versions of spacetime based solely on the two-point Lorentzian distance. We emphasize that, as noted in previous works, this structure is essentially unique if minimal conditions of continuity, and relative compactness of diamonds are imposed - the relative compactness of diamonds being natural for justifying the interpretation of these spaces as low regularity versions of globally hyperbolic spacetimes. In this work, in confirmation of this interpretation, we prove that every Lorentzian metric space admits a Cauchy time function.
}



\maketitle

%\begin{keyword}


%Berwald space \sep Minkowski space \sep Ricci scalar.

%{Keywords: Finsler gravity, Finsler spacetime, energy-momentum conservation, Finsler connection, non-linear connection.}

%\end{keyword}

%\newpage

\section{Introduction}
The present work is an introduction to the theory of {\em Lorentzian metric spaces} as developed in previous work in collaboration with S. Suhr \cite{minguzzi22} (bounded case) and A. Bykov and S. Suhr \cite{minguzzi24b} (unbounded case).

As an original contribution, we present a new result on the existence of Cauchy time functions for these spaces, thus confirming their interpretation as low regularity versions of globally hyperbolic spacetimes. This proof provides an affirmative answer to a question posed by L. Garc{\'\i}a-Heveling at the BIRS-IMAG  meeting in Granada.\footnote{``Geometry, Analysis, and Physics in Lorentzian Signature", May 4 - 9, 2025}

The traditional approach to proving this type of result uses Geroch's volume function, which typically involves subtle issues connected to showing that the boundaries $\partial I^\pm(p)$, for every spacetime point $p\in M$, have vanishing measure \cite{geroch70,hawking73,dieckmann88}.  Burtscher and García-Heveling recently followed a volume-type approach in a low regularity Lorentzian length space setting à la Kunzinger-S\"amann \cite{kunzinger18,burtscher25}.

Interestingly, our proof is extremely simple and appears to be new already in the smooth case. It is simpler than Geroch's volume function proof \cite{geroch70,hawking73} as it requires no spacetime measure. Instead, all Lorentzian metric space axioms enter the proof neatly.

This line of research is motivated by the long-standing interplay between physics and geometry, which suggests the need for a theory of Lorentzian geometry that can incorporate discreteness and non-regularity. From quantum mechanics to the quest for a theory of quantum gravity, the fundamental description of spacetime at high energies may require a departure from the smooth manifold paradigm. Concurrently, developments in metric geometry, such as the theory of Alexandrov spaces and Ricci limit spaces based on Gromov-Hausdorff convergence, have shown the power of synthetic and discrete approaches in the positive signature setting.

A natural question is whether these mathematical advances can be translated to the Lorentzian setting, which is inherent to General Relativity. Unfortunately, this translation is far from straightforward. Fundamental difficulties arise: Lorentzian balls are non-compact, and the Lorentzian distance $d$ \cite{beem96,minguzzi18b} satisfies a \emph{reverse triangle inequality} ($d(p,r) \geq d(p,q) + d(q,r)$ for $p \ll q \ll r$) and hence has a non-trivial relationship with the topology of the space.

Synthetic, namely coordinate free, approaches to relativity are actually as old as the theory itself. The very first papers by Einstein on Special Relativity were synthetic in approach. They used gedanken experiments involving reflection of light beams over mirrors, employing the typical argumentative structure of Euclidean geometry with its points and lines. Not surprisingly, axiomatizations of relativity based on synthetic concepts were developed early on \cite{robb14}.

Later work, such as that by Busemann \cite{busemann67}, or the more recent work by Kunzinger-S\"amann \cite{kunzinger18}, laid important groundwork, reframing this type of investigation in a more modern mathematical language. However, these works rely on auxiliary structures, such as a pre-existing topology \cite{busemann67} or an auxiliary  metric \cite{kunzinger18}, which lack a clear Lorentzian justification. Our objective in \cite{minguzzi22} was to develop a purely metric approach based \emph{only} on the Lorentzian distance function $d$. Our guiding principle was the identification of definitions that are stable under a suitable notion of Gromov-Hausdorff (GH) convergence, ensuring that the resulting theory is natural and robust.

In our first paper \cite{minguzzi22}, we introduced \emph{Bounded Lorentzian Metric Spaces} (BLMS), where the entire space satisfies a mild type of compactness condition which implies a bounded timelike diameter. In the present paper we include a discussion of the subsequent development in the \emph{unbounded} case as developed in \cite{minguzzi24b}. We recall the definition and fundamental properties of Lorentzian metric spaces, the notion of GH-convergence for them, and emphasize some central properties like the fact that (pre)length spaces are  stable under GH-limits.

%Since GH-convergence was one of our guiding principles, it is worth mentioning that approaches to GH-convergence for Lorentzian spacetimes were proposed before our work by Noldus \cite{noldus04}, Bombelli-Noldus \cite{bombelli04}, Sormani-Vega \cite{sormani16}, Allen-Burtscher \cite{allen19}, at the same time by M\"uller \cite{muller22b}, and subsequently by Cavalletti-Mondino \cite{cavalletti20}, Mondino-S\"amann \cite{mondino25}, and Che-Perales-Sormani \cite{che25}.

As mentioned, our work was guided by the study of GH-convergence for spacetime, a concept for which some approaches already existed. Preceding our work were those of Noldus \cite{noldus04}, Bombelli-Noldus \cite{bombelli04}, Sormani-Vega \cite{sormani16}, and Allen-Burtscher\cite{allen19}, with Müller's approach \cite{muller22b} being developed concurrently. Subsequent contributions include those by Cavalletti-Mondino \cite{cavalletti20}, Mondino-S{\"a}mann \cite{mondino25}, and Che-Perales-Sormani  \cite{che25}.

Our approach had the advantage of establishing a direct connection with causets, i.e.\ finite weighted oriented graphs at the foundation of  Causal Set Theory \cite{surya19}. In our framework, causets are just finite BLMS, and interestingly the family of causets is GH-dense in the family of the general BLMS.\footnote{ The recent work \cite{braun25} on the spacetime reconstruction problem has deepened the link of BLMS to Causal Set Theory.} It is precisely this property that allowed us to obtain an analog of Gromov's precompactness theorem for BLMS.

Still, we stress that our most significant result was not in the direction of the precompactness theorem,
% as the version we obtained certainly would require further improvements lacking a clear connection with curvature conditions.
%In fact, it uses an assumption on the uniform bound of $\epsilon$-nets needed to cover the elements of the sequence of spaces $X_n$ that, although analogous to similar bounds in the standard metric theory, has yet to be connected in the Lorentzian theory to curvature conditions.
rather it was  in  the identification of a minimal and particularly convenient notion of Lorentzian metric/length space, where the  Length Space property is also stable under GH-convergence. The merit of our work was  to show that a simple geometry of Lorentzian metric spaces based solely on the Lorentzian distance $d$ is possible, essentially solving the difficulties in connecting $d$ with other apparently primitive concepts such as topology and causal structure. This issue was left open by previous approaches  as soon as they included in their framework
auxiliary concepts that do not seem to have observational evidence (note that working with the Lorentzian distance is reasonable as it corresponds to the Lorentzian metric $g$  which is a field concept).

\section{Definition and first properties of  LMS}

Given a function $d: X \times X \to [0, \infty)$ let
\begin{align*}
I&=\{(x,y)\in X^2: d(x,y)>0\},\\
 I_\epsilon &= \{(x,y) \in X^2 : d(x,y) \geq \epsilon\}
\end{align*}
  and let $I(x,y) = I^+(x) \cap I^-(y)$ where $I^+(x)=\{y: (x,y)\in I \}$,  $I^-(y)=\{x: (x,y)\in I \}$.
The core object of our study is defined as follows.

\begin{definition}[Lorentzian Metric Space (LMS)] \label{defl}
A \emph{Lorentzian metric space} $(X, d)$ is a set $X$ endowed with a function $d: X \times X \to [0, \infty)$, called the \emph{Lorentzian distance}, satisfying:
\begin{enumerate}
    \item[(i)] Reverse Triangle Inequality: For every $x, y, z \in X$ with $d(x, y) > 0$ and $d(y, z) > 0$, we have $d(x, z) \geq d(x, y) + d(y, z)$.
    \item[(ii)] Topology and Compactness: There exists a topology $\mathcal{T}$ on $X$ such that $d$ is continuous, and for every $x, y \in X$ and every $\epsilon > 0$, the set $I_\epsilon \cap (\overline{I(x,y)} \times \overline{I(x,y)})$ is compact.
    \item[(iii)] (weak) Distinction: $d$ distinguishes points: for every pair $x \neq y$, there exists $z$ such that $d(x,z) \neq d(y,z)$ or $d(z,x) \neq d(z,y)$.
\end{enumerate}
\end{definition}

The relation $I$ is the  chronological relation, also denoted $\ll$, while $I(x,y)$ is the chronological diamond. The chronological future of a set $S$ is as usual denoted with $I^+(S)$, while the past is denoted with $I^-(S)$. The chronologically convex hull is $I(S):=I^+(S)\cap I^-(S)$.

 A  \emph{Bounded LMS} (BLMS) is a LMS for which  the sets $I_\epsilon$ themselves are compact for all $\epsilon>0$; this implies $d$ is bounded.

The future chronological  boundary $X^+$ consists of all points with an empty chronological future. Similarly,  $X^-$ consists of all points with an empty chronological past. The condition of both boundaries being empty reads $I(X)=X$.

The spacelike boundary is $i^0=X^+\cap X^-$ and consists of at most a single point (by (iii)). In the BLMS case it is always possible to adjoin $i^0$ (or remove it) this process being essentially the one-point compactification.

This spacetime definition is minimal and natural. Smooth globally hyperbolic spacetimes \cite{hawking73} and causal sets (causets) \cite{surya19} are examples of LMS.

Condition (ii) balances the topology of the spacetime. Specifically, a finer topology (more open sets) makes it easier for
$d$ to be continuous, but harder to satisfy the compactness condition.  It is essentially a formulation of global hyperbolicity, as becomes apparent for spacetimes without a chronological boundary, as we have

\begin{theorem}\label{cg-lms}
   A Lorentzian metric space $(X,d)$ such that $I(X)=X$ is characterized by  the following properties:
    \begin{enumerate}
        \item[(i)] $d$ satisfies the reverse triangle inequality;
        \item[(ii)] There is a topology $T$ on $X$ such that $d$ is continuous in the product topology and, for every $x,y\in X$, $\overline{I(x,y)}$ is compact;
        \item[(iii)] $d$ distinguishes points;
        \item[(iv)] $I(X)=X$.
    \end{enumerate}
\end{theorem}

%This result clarifies that, as we noted in our first work \cite{minguzzi22}, there is a certain uniqueness and inevitability of the notion of LMS. Any  abstract formulation of Lorentzian geometry leads to LMS unless

This result clarifies that, as we noted in our first work \cite{minguzzi22}, the definition of an LMS possesses a notable uniqueness and naturalness. %Within a purely distance-based framework, any
Alternative abstract formulation of Lorentzian geometry will be closely related to LMS unless
\begin{itemize}
\item[(a)] they use ingredients different from $d$ alone (an auxiliary metric, an auxiliary closed order, an auxiliary topology, etc.),
\item[(b)] the assumptions imposed do not reproduce global hyperbolicity as they do not imply good properties for $d$.
\end{itemize}
Considering a two-point function with value in $\{-\infty\}\cup \mathbb{R}$ rather than $\mathbb{R}$ does not change this conclusion as discussed at length in \cite{minguzzi24b} and briefly recalled in what follows.

In the remainder of this introduction we shall always assume $I(X)=X$, as otherwise some results become a bit more technical. This excludes finite causets which can
be better treated as BLMS anyway.

%as we shall recall shortly.
%Indeed, if the space has no past or future boundary ($X^- = \emptyset$ or $X^+ = \emptyset$), it is equivalent to the relative compactness of all chronological diamonds $I(x,y)=I^+(x)\cap I^-(y)$.

As we shall see, another crucial concept is that of a \emph{generating set} $\mathcal{G} \subset X$, for which $X = I(\mathcal{G})$. If a countable generating set exists, the LMS is called \emph{countably generated}. A \emph{sequenced LMS} $(X, d, (p^k)_{k\in\mathbb{N}})$ is an LMS with a distinguished generating sequence.

\begin{remark}
Condition (iii) is not so restrictive. If (i) and (ii) are satisfied it is often possible, e.g.\ for bounded Lorentzian metric spaces, to accomplish (iii) passing to a {\em distance quotient} which we introduced as follows: $x\sim y$ if they are not distinguished by $d$. The Lorentzian distance then passes to the quotient (see also Thm.\ \ref{tjg} and relative discussion).
\end{remark}

\section{The canonical topology}
\label{sec:topology}

A priori, Definition \ref{defl} requires the existence of \emph{some} topology. A central result is that for an LMS with no chronological boundary, this topology is, in fact, unique and can be constructed directly from $d$. The same result holds for a BLMS, at least for the topology induced on the complement of $i^0$. When a BLMS contains $i^0$ the canonical topology is the coarsest  one among those that satisfy (ii) (it can be shown to exist), which fixes also the neighborhood system of $i^0$.

The proof is based on the following topological result \cite{minguzzi22} which we could not find in the literature

\begin{lemma}\label{lem:subbasis}
    Let $X$ be a topological space, and let $\mathscr{A}$ be a family of open subsets of $X$ such that
     \begin{itemize}
     \item[(a)] $\mathscr{A}$-Hausdorffness: for any $x,y\in X$, $x\ne y$, there are $A,B\in \mathscr{A}$ satisfying $x\in A$, $y\in B$, $A\cap B=\varnothing$,
     \item[(b)] $\mathscr{A}$-local compactness: for any $x$ there is $U\in \mathscr{A}$ such that $x\in U\subset C$, where $C$ is compact.
     \end{itemize}
     Then $\mathscr{A}$ is a subbasis for the  topology.
\end{lemma}

The idea in the LMS case is to experiment with some families derived from $d$ to check if they satisfy conditions (a) and (b) and hence provide a subbasis for the topology (the chronological diamonds as in Alexandrov topology are not enough unless additional conditions recalled below are imposed, cf.\ Cor.\ \ref{oddr}).
By using the above Lemma we were able to prove

\begin{theorem}[Uniqueness and Characterization of the Topology]
Let $(X, d)$ be a Lorentzian metric space with no chronological boundary ($I(X)=X$). Then, a topology $\mathcal{T}$ satisfying property (ii) is Hausdorff, locally compact, and unique. Moreover, the sets of the form
\begin{equation}
\{ q : a < d(p, q) < b \} \cap \{ q : c < d(q, r) < e \}
\end{equation}
with $p, r \in X$ and $a, b, c, e \in \mathbb{R} \cup \{-\infty, \infty\}$, form a subbasis for $\mathcal{T}$.
\end{theorem}

This unique topology, called {\em the LMS topology}, and denoted $\mathscr{T}$, is really the initial topology induced by the family of functions $\{d_p, d^p : p \in X\}$, where $d_p(\cdot)=d(p,\cdot)$ and $d^p(\cdot)=d(\cdot,p)$. An immediate corollary that can be deduced from this fact is that isometries (bijective distance-preserving maps) are homeomorphisms.

\begin{corollary}\label{crl:isometry-isomorphism}
    Let $(X,d)$ and $(Y,d)$ be two Lorentzian metric spaces without  chronological boundary and let $\phi:X\to Y$ be a bijective distance-preserving map. Then $\phi$ is a topological homeomorphism.
\end{corollary}

An interesting observation is

\begin{theorem}
The basis elements of $\mathscr{T}$ constructed from $d$ as above are chronologically convex.
\end{theorem}

The following provides a condition which ensures the equivalence between LMS topology and Alexandrov topology (which is that generated by chronological diamonds $I(p,q)$)

\begin{corollary} \label{oddr}
    Let $(X,d)$ be a Lorentzian metric space. Assume that for every $x\in X$ we have %and open neighborhood $U\ni x$,  the sets $U\cap I^{+}(x)$ and $U\cap I^{-}(x)$ are not empty (in other words,
    $x \in \overline{I^\pm (x)}$. Then the LMS topology coincides with the Alexandrov topology.
\end{corollary}

%Furthermore, if every point $p$ satisfies $p \in \overline{I^{\pm}(p)}$, the \emph{Alexandrov topology} (generated by chronological diamonds $I(p,q)$) coincides with the metric topology $\mathcal{T}$.

%Part of our work on LMS was devoted to establishing to what extent an LMS resembles a BLMS on compact subsets. The results obtained imply that a LMS inherits the properties of the BLMS, where the latter  objects are generally easier to deal with due to their better compactness properties.

Part of our work on LMS involved determining how closely an LMS resembles a BLMS on compact subsets. Our results show that an LMS inherits the properties of a BLMS. The latter are generally easier to analyze due to their stronger compactness properties. This was essentially the strategy through which the following result was obtained

\begin{theorem}
    \label{prop:cg-Polish}
    If $(X,d)$ is a countably generated Lorentzian metric space, then its LMS topology is $\sigma$-compact, second-countable and Polish.
\end{theorem}

This result is extremely useful as it allows one to apply to LMS the theory of optimal transport, see e.g.\ \cite{braun23b}.

For a BLMS a possible metric that induces the topology $\mathscr{T}$ is the {\em distinction metric}
\begin{equation}
\gamma(x,y)=\max\left(\sup_{z\in X} \vert d(x,z)-d(y,z) \vert, \ \sup_{z\in X} \vert d(z,x)-d(z,y) \vert  \right).
\end{equation}
which measures how much a point $x\in X$ can be distinguished from a point $y\in X$ by using the function $d$.

The {\em Noldus  (strong) metric} is defined as (originally introduced for manifolds)
\begin{equation}
D(x,y)=\sup_{z\in X} \vert d(z,x)+d(x,z)-d(z,y)-d(y,z)\vert .
\end{equation}
For BLMS we proved that the two metrics coincide \cite{minguzzi22}.

The distinction metric has an interesting interpretation that elucidates its properties. Let us consider a
 topological space $(X,\mathscr{T})$ endowed with a continuous bounded function $d\colon X\times X\to [0,\infty)$. The bounded function $\gamma: X\times X \to [0,+\infty)$ defined as above satisfies (a) triangle inequality (b) $\gamma(x,y)=\gamma(y,x)$, namely it is a pseudo-metric (as we are not assuming the distinction property).

For $x\in X$ and $\epsilon>0$ we denote the  distinction ball as follows
\[
B_\epsilon^\gamma(x)=\{p\in X: \gamma(x,p)<\epsilon\}
\]

Now, consider the space $C^0(X,\mathbb{R})\times C^0(X,\mathbb{R})$ endowed with the topology of uniform convergence, namely that induced by the norm
\[
 \Vert (f,g)\Vert=\max\{\sup_ z\vert  f(z) \vert, \sup_ z\vert  g(z) \vert\}.
\]
The associated metric is
\[
\textrm{dist}_\infty((f,f'),(g,g')):=\max\{\textrm{dist}_\infty
(f,g),\textrm{dist}_\infty (f',g')\},
\]
where
\[
\textrm{dist}_\infty(f,g)=\Vert f-g\Vert_\infty=\sup_z \vert f(z)-g(z)\vert.
\]

Consider the Kuratowski-type  map
\begin{align*}
\kappa\colon X &\to C^0(X,\mathbb{R})\times C^0(X,\mathbb{R}) \\
x& \mapsto (d_x,d^x)
\end{align*}
Since $d$ is bounded, the image of $\kappa$ actually lies in the subspace of bounded continuous functions. Observe that the distinction metric is related to the Kuratowski map as follows
\begin{equation}
\gamma(x,y)=(\kappa^*\textrm{dist}_\infty)(x,y)=\textrm{dist}_\infty(\kappa(x),\kappa(y))
\end{equation}
and so that $\gamma$ is continuous if $\kappa$ is continuous.

We have the following result whose proof is found in the proof of \cite[Prop.\ 1.9]{minguzzi22}

\begin{theorem} \label{tjg}
Let $(X,\mathscr{T})$ be a topological space endowed with a  $\mathscr{T}\times \mathscr{T}$-continuous function $d\colon X\times X\to [0,\infty)$ and suppose that the sets $I_\epsilon:=\{(x,y): d(x,y)\ge \epsilon\}$, $\epsilon>0$, are compact (hence $d$ is bounded). Then $\kappa$ is continuous and $\gamma$ is a continuous pseudo-metric. In particular, the distinction balls are $\mathscr{T}$-open.
%If additionally, the distinction property (iii) holds, then $\kappa$ is a homeomorphism onto the image.
\end{theorem}
Observe that the first statement uses as assumption condition (ii) of Def.\ \ref{defl} in the version for BLMS \cite[Def.\ 1.1]{minguzzi22}, and the fact that it ensures openness of distinction balls can be used to show that under (ii) the function $d$ passes  under the distance quotient  to a {\em continuous} function $\tilde d$ (continuity is obtained passing to the quotient the inequality $\vert d(p', q')- d(p,q)\vert\le  \vert d(p', q')- d(p',q)\vert +\vert d(p',q)- d(p,q)\vert< \epsilon$ for  $p'\in  B^\gamma_{\epsilon/2}(p)$, $q'\in  B^\gamma_{\epsilon/2}(q)$). This implies that for a BLMS condition (iii) is not really so demanding as it can be recovered assuming just (ii) and performing a distance quotient.

The distinction metric is certainly a very natural tool but is  not so essential in the study of LMS.
In general, metrics are not so important, in fact, passing from  BLMS to LMS it becomes clear that they are non-canonical. They will be replaced with the notion of uniformity that we shall recall later on.

%For countably generated spaces, the topology has very strong properties: it is $\sigma$-compact, second-countable, and Polish. This is fundamental for applications, e.g., in enabling the use of optimal transport theory.

\section{The canonical causal relation}
\label{sec:causal}

From $d$ it is easy to define $I$, we just set $I=\{(x,y)\vert \  d(x,y)>0\}$. We have now the problem of defining also the  causal relation $J$. The problem is solved by considering the  largest relation compatible with the reverse triangle inequality.

%\begin{block}{}
%Answer: as the largest relation compatible with the reverse triangle inequality.
%\end{block}

%From the chronological relation $\ll$ defined by $d$, we can recover a causal relation $J$ (denoted $\leq$) that is closed, reflexive, transitive, and extends $\ll$.
%This is done as follows:
\begin{definition}[Causal Relation]
For a pair $(X, d)$  where $d: X \times X\to [0,+\infty)$ satisfies the reverse triangle inequality, the \emph{(extended) causal relation} $J \subset X \times X$ is defined by:
\begin{equation}
J = \big\{ (x, y) \in X \times X \mid d(p, y) \geq d(p, x) \text{ and } d(x, p) \geq d(y, p),\ \forall p \in X \big\}.
\end{equation}
\end{definition}
It is also denoted $\le$ while $J\backslash \Delta$ is also denoted $<$.
For $J$ we get  all the desirable properties. We recall that a relation $R\subset X\times X$ is antisymmetric if $(x,y)\in R$ and $(y,x)\in R$ implies $x=y$.
\begin{theorem} \label{thm:Jproperties}
     The relation $J$ is closed, reflexive, transitive and antisymmetric. Moreover, $I\subset J$ and
     \[
     I\circ J\cup J\circ I\subset I \quad \textrm{(push up/ Kronheimer and Penrose's causal space property)}.
     \]
     If $(x,y), (y,z)\in J$ then
     \begin{equation}
     d(x,y)+d(y,z)\le d(x,z).
      \end{equation}
      %If  $(X,d)$ distinguishes points (property (iii) of LMS)  then $J$ is antisymmetric.
\end{theorem}
 This construction allows us to characterize our spaces in terms of familiar causal concepts.

To start with, there is a connection of LMS with global hyperbolicity in the sense of topological ordered spaces
\begin{theorem}\label{thm:lms-via-emeralds}
        Let $(X,d)$ be a Lorentzian metric space without chronological boundary, and $C\subset X$ any compact subset. Let $K$ be a closed order on $X$ such that $I\subset K\subset J$. Then the set $K(C)$ is compact.
    \end{theorem}

    The typical choice would be  $K=J$.
In some applications one might want to work with a different closed relation. Still by imposing suitable properties we are back to LMS.
\begin{theorem}[Equivalence with Global Hyperbolicity] \label{ocaer} $\empty$\\
   A Lorentzian metric space $(X,d)$ such that $I(X)=X$ is characterized by  the following properties:
    \begin{enumerate}
        \item[(i)] $d$ satisfies the reverse triangle inequality;
        \item[(ii)] There is a topology $T$ on $X$ and a closed order $K$, $I\subset K\subset J$ such that $d$ is continuous in the product topology and, for every  compact set $C$, $K(C)$ is compact;
        \item[(iii)] $d$ distinguishes points;
        \item[(iv)] $I(X)=X$.
    \end{enumerate}
\end{theorem}
This result for $K=J$ establishes that LMS indeed are characterized by properties that coincide with those of global hyperbolicity in the smooth setting, this time formulated via the causal relation. It confirms that our axioms are equivalent to a robust, abstract notion of global hyperbolicity.

The previous result preserves its validity replacing $C\subset X$ with $\{p,q\}, p,q\in X$ thus leading to a more traditional `causal diamonds' formulation.

It is possible to express the whole theory of LMS  using a function $l: X\times X \rightarrow { \{-\infty\}}\cup [0,+\infty)$ instead of  $d$, where $l$ satisfies an extended reverse triangle inequality (i.e.\ holding for every triple $x,y,z\in X$). A number of results clarify the translation between the two choices \cite{minguzzi24b}. For instance, for what concerns the continuity properties

\begin{lemma}\label{lem:braun-d}
    Let $X$ be a topological space, $l: X\times X \rightarrow \{-\infty\}\cup [0,+\infty)$ an arbitrary function.
    Then the following conditions are equivalent:
    \begin{enumerate}
        \item $l$ is upper semi-continuous, and $l_{+}=\max(0,l)$ is lower semi-continuous;
        \item $l_+$ is continuous and the set
        \[
        K_{l}:=\{(x,y)\in X\times X: l(x,y)\geq 0\}
        \]
        is closed.
    \end{enumerate}
\end{lemma}

Of course, $d$ would correspond to $l_+$. The rationale of this type of results, discussed at length in \cite{minguzzi24b}, is that using $l$ means using $d$  in addition to a special selection for an intermediate closed relation $I\subset K_l\subset J$. However, as soon as one imposes on $K_l$ a condition on the compactness of diamonds, by Theorem \ref{ocaer}, one falls in the category of LMS.

A framework based on $l$ was used by Braun and McCann in \cite{braun23b} in a broad optimal transport study. The non-compact spaces introduced there were also constructed relying on the notion of BMLS (so inheriting its good properties such as having a Polish topology or sharing the validity of limit curve theorems). Once some additional local (GH-unstable) conditions are removed, these spaces turn out to be  LMS in our sense \cite{minguzzi24b}.

% by McCann for optimal transport on Lorentzian manifolds. Braun and McCann (arXiv:2312.17158) have a non-compact low regularity framework based on BLMS similar to that of LMS but expressed via $l$.

%A more precise dictionary is in the second paper.

%\begin{block}{}
%This function was used by McCann for optimal transport on Lorentzian manifolds. Braun and McCann (arXiv:2312.17158) have a non-compact low regularity framework based on BLMS similar to that of LMS but expressed via $l$.
%\end{block}

Concerning local conditions,  we investigated the following  interesting property.
In a LMS, in general, we have  $\bar I \subset J$  and $\mathscr{A} \subset \mathscr{T}$ where $\mathscr{A}$ is the Alexandrov topology. The equalities do not necessarily hold (the strict inclusion $\bar I \subsetneq J$ would be referred as presence of   {\em causal bubbles} \cite{chrusciel12,minguzzi17,grant20,heveling22}).
%For the equality to hold we need to assume some local condition.
We obtained the following result which extends Cor.\ \ref{oddr}

\begin{theorem}[No Gaps Theorem]
Let $(X,d)$ be a  Lorentzian metric space $(X,d)$ such that $I(X)=X$.
If for every point, $p\in \overline{I^\pm(p)}$, then there are no gaps: $\bar I = J$ and  $\mathscr{A} = \mathscr{T}$.
\end{theorem}

Unfortunately, the condition $p\in \overline{I^\pm(p)}$ seems unlikely to be preserved under GH-limits. This is connected to a behavior that can already be observed  in cone structures (where causality derives from a distribution of cones over a smooth manifold \cite{minguzzi17}). The cones can collapse in the limit to cones with non-empty interior so the limit point in the limit space might have an empty chronological future.

%Example of a proper Lorentz-Finsler space collapsing to a closed Lorentz-Finsler space.

\section{Isocausal curves and (pre)length spaces}
\label{sec:length}

%We define the length of causal curves and the notion of a Lorentzian (pre)length space within our framework.

A continuous curve $\sigma : [a, b] \to X$ is \emph{isocausal} if $a \leq s < t \leq b$ implies $\sigma(s) < \sigma(t)$.
%Its \emph{Lorentzian length} is defined by
%\begin{equation}
%L(\sigma) = \inf \sum_{i=0}^{k-1} d(\sigma(t_i), \sigma(t_{i+1})),
%\end{equation}
%where the infimum is over all partitions of $[a,b]$.
An isocausal curve is \emph{maximizing} (or {\em maximal})   if for $a\le t<t'<t''\le b$ it satisfies
\begin{equation}
d(\sigma(t),\sigma(t'))+d(\sigma(t'),\sigma(t''))=d(\sigma(t),\sigma(t'')).
\label{eqn:maxcurv}
\end{equation}

\begin{definition}[Lorentzian (pre)length space]
A LMS $(X, d)$ is a:
\begin{itemize}
    \item[(a)] \emph{prelength space} if for every $x \ll y$, there exists an isocausal curve connecting them.
    \item[(b)] \emph{length space} if for every $x \ll y$, there exists a maximizing isocausal curve connecting them.
\end{itemize}
\end{definition}

\begin{remark}
A feature of LMS that must be taken into consideration is that a maximizing curve might have {\em null segments}, that is pair of points in the interior at zero Lorentzian distance though the endpoints are chronologically related.
\end{remark}

A key tool in the study of isocausal curves is the existence of \emph{time functions}, a property also known in the smooth category as {\em stable casuality}. We recall that a time function is a continuous function $\tau$ such that $x< y \Rightarrow \tau(x) <\tau(y)$. For any countably generated LMS, a bounded time function $\tau: X \to [-1,1]$ (continuous and strictly increasing on isocausal curves) can be constructed explicitly from $d$:
\begin{equation}
\tau(x) = \sum_{n=1}^{\infty} \frac{1}{2^n} \left( \frac{d(x_n, x)}{1 + d(x_n, x)} - \frac{d(x, x_n)}{1 + d(x, x_n)} \right),
\end{equation}
where $\{x_n\}$ is a dense sequence.

Note that $\{x_n\}$ distinguishes points in the sense of Def.\ \ref{defl}. Indeed, if $x,y\in X$, $x\ne y$, there is $z\in X$ such that $d(x,z)\ne d(y,z)$ or  $d(z,x)\ne d(z,y)$. By continuity of $d$ these inequalities in $z$ remain valid  in a neighborhood of $z$ where we can find some $x_n$ to replace $z$ (see also the proof of \cite[Prop.\ 1.11]{minguzzi22}. This distinguishing property is what ensures the strict inequality in the right-hand side of  $x< y \Rightarrow \tau(x) <\tau(y)$.
%This distinguishing property is what ensures th

Given a time function $\tau$ any  isocausal curve $\sigma$ admits a {\em $\tau$-uniform reparametrization} $\tau(\sigma(s))=s$.
%This is the parametrization used in limit curve theorems.
Time functions allow for a canonical parametrization of curves and are thus crucial for proving limit curve theorems.

In this connection the main idea that we exploit is that our parametrization comes from a time function and hence ultimately from $d$. This is of extreme importance as, thanks to this fact, the isocausal curves can be controlled under GH-limits. For comparison, in Kunzinger-S{\"a}mann approach to length spaces \cite{kunzinger18} one would introduce an auxiliary metric and consider the family of Lipschitz curves to use Ascoli-Arzel\`a and get a limit curve theorem. In our approach this is not  needed but a larger family of curves is used.

Indeed, our curves cannot be compared with those of KS because they impose rectifiability with respect to some metric - an ingredient which is absent in our approach. Our family of curves is thus considerably larger. There are also differences on causality conditions such as non-total imprisonment but, luckily, they do not affect the stronger causality condition: global hyperbolicity \cite{minguzzi23}.

\begin{remark}
In a BLMS the isocausal curves are non-rectifiable with respect to the distinction metric \cite{noldus04b,muller19,minguzzi22} which goes to show that this metric is less useful than one might  expect.
\end{remark}

Before coming to the  limit curve theorem, let us mention some preliminary results. First of all we were able to prove the following Lemma which appears to be  new also in the smooth Lorentzian manifold case \cite{minguzzi24b}. It is used to study convergence of isocausal curves on a countable subset of their domain, to then infer existence of the limit curve on the whole interval.
\begin{lemma}\label{lem:curves-from-dense}
     Let $(X,d)$ be a countably-generated Lorentzian metric space, and let $\tau: X\to \mathbb{R}$ be a time function. Let $a,b\in \mathbb{R}$  be such that $a<b$, and assume that $Q$ is a dense subset of $[a,b]$ such that $a,b\in Q$. Let $\zeta: Q\to X$ be such that
    \begin{itemize}
        \item $\tau(\zeta(t))=t$, $\forall t\in Q,$
        \item  $\zeta(t)\leq \zeta(t')$ whenever $t,t'\in Q$, $a\leq t <t' \leq b$.
    \end{itemize}
    Then $\zeta$ extends to an isocausal curve $\overline{\zeta}:[a,b]\to X$ parametrized by $\tau$. Moreover, if $\zeta$ satisfies the maximizing condition, then $\overline{\zeta}$ is maximizing.
\end{lemma}
The important feature we got here is the continuity property of $\overline{\zeta}$ which is included in the definition of isocausal curve. For instance, as a corollary we get\footnote{Remember that a map $x: [a,b]\to X$ is isotone if $t\le t'$ implies $x(t)\le x(t')$.} for   $Q=[a,b]$

\begin{corollary}
An isotone map $x: [a,b]\to X$ parametrized by a time function (i.e. any map satisfying the two conditions above) is necessarily continuous, and thus an isocausal curve.
\end{corollary}

Since the geometry is controlled uniquely by the Lorentzian distance, we expect that convergence of isocausal curves should be expressible through it. In fact, this is possible as shown by the following result \cite{minguzzi24b}
\begin{theorem}\label{thm:d-uni}
%    Let $\gamma$ be any complete metric of a Lorentzian metric space $(X,d)$ inducing its topology,
    Let us consider a countably generated Lorentzian metric space $(X,d)$,
    let $K$ be a compact subset of $X$, let $\sigma_n: [a_n,b_n]\to K$ be a sequence of continuous curves, and let $\sigma: [a,b]\to K$ be yet another continuous curve. The sequence $\sigma_n$ uniformly converges to $\sigma$ iff for each $z\in X$ the sequences of functions  $d_z\circ \sigma_n$ and $d^z\circ \sigma_n$   converge uniformly to $d_z\circ \sigma$ and $d^z\circ \sigma$, respectively.
\end{theorem}
Ultimately, this theorem is a consequence of the existence of a canonical quasi-uniformity which we shall recall later on.

An important surprising consequence is

\begin{corollary}
Pointwise convergence of the sequence of isocausal curves $\sigma_n$ to the isocausal curve $\sigma$ implies uniform convergence.
\end{corollary}

We are ready to state the limit curve theorem

\begin{theorem}[Limit curve theorem]\label{thm:lim-curv}
    Let $(X, d)$ be a Lorentzian metric space without chronological boundary. Let $\sigma_n: [a_n, b_n] \to X$ be a sequence of isocausal curves parametrized with respect to a given time function $\tau$, $\tau(\sigma_n(t))=t$. Suppose that
    \[
        \lim_{n\to \infty }\sigma_n(a_n)=x, \          \lim_{n\to \infty }\sigma_n(b_n)=y,
    \]
    and $x\neq y$. Then there exists a $\tau$-uniform isocausal curve $\sigma : [a, b] \to X$ and a subsequence $\{\sigma_{n_k}\}_k$ that converges uniformly to $\sigma$. If the curves $\sigma_n$ are maximizing then so is $\sigma$.
\end{theorem}

%\begin{block}{Lorentzian length}
Let $\sigma: [0,1]\to X$ be an  isocausal curve on a LMS. Its  \emph{Lorentzian length} is
\begin{equation}
L(\sigma):=\textrm{inf} \sum_{i=0}^{k-1} d(\sigma(t_i),\sigma(t_{i+1})) ,
\end{equation}
where the infimum is over the set of all partitions.
%\end{block}
By the reverse triangle inequality $L(\sigma)\le d(\sigma(0), \sigma(1))$ and equality holds iff it is maximizing.

Standard properties of the Lorentzian length functional are preserved for LMS

%\begin{proposition}
%Let $\sigma: [0,1]\to X$ be an isocausal curve with  endpoints $x$ and $y$.
%We have $L(\sigma)=d(x,y)$ iff $\sigma$ is maximizing.
%\end{proposition}

\begin{theorem}[Upper semi-continuity of the length functional] $\empty$ \\
Let $(X,d)$ be a Lorentzian metric space. Let $\sigma_n: [0,1]\to X$ and $\sigma: [0,1]\to X$ be  isocausal curves and suppose that $\sigma_n\to \sigma$ pointwise. Then
\begin{equation}
\limsup L(\sigma_n)\le L(\sigma).
\end{equation}
\end{theorem}

Among the nice properties of $L$ there is the alternative characterization of the length space property.
Let $(X,d)$ be a Lorentzian prelength space. Then for $x\ll y$ we can define
\begin{equation}
\check{d}(x,y)=\sup_{\sigma(0)=x,\sigma(1)=y} L(\sigma),
\end{equation}
where the supremum goes over isocausal curves. If $x\ll y$ fails, we define $\check{d}(x,y)=0$. Note that  $\check{d}\le d$.

\begin{theorem}\label{thm:alt-lenspace}
Let $(X,d)$ be a countably generated Lorentzian prelength space. The following conditions are equivalent
\begin{itemize}
\item[(i)] $ \check d=d$,
\item[(ii)] $(X,d)$ is a length space.
\end{itemize}
\end{theorem}

We were also able to define sectional curvature bounds on LMS, but those will not be recalled here.

%\section{Limit curve theorem}

%We establish a powerful \emph{Limit Curve Theorem}, guaranteeing that a limit of a sequence of isocausal curves is itself an isocausal curve, and maximality is preserved in the limit.

\section{Gromov-Hausdorff convergence}
\label{sec:gh}

Since our study of BLMS we found  convenient to define the Gromov-Hausdorff convergence of spacetimes via the notion of correspondence rather than through embeddings. This appeared to be the most natural choice also in view of optimal transport theory, where couplings are measures on the Cartesian product
$X\times Y$.

\begin{definition}
Let $X$ and $Y$ be sets.
A relation $R\subset X\times Y$ is a {\it correspondence} if
\begin{itemize}
\item[(i)]  for all $x\in X$ there exists $y\in Y$ with $(x,y)\in R$, and
\item[(ii)]  for all $y\in Y$ there exists $x\in X$ with $(x,y)\in R$.
\end{itemize}
\end{definition}

A {\em bounded space} is a set with a bounded function  $d_X\colon X\times X\to [0,\infty)$.

\begin{definition}
Let  $(X,d_X)$ and $(Y,d_Y)$ be bounded spaces and $R\subset X\times Y$ be a
correspondence.
The {\it distortion} of  $R$ is defined as:
\begin{equation}
\textrm{dis}\, R:=\sup\{\vert d_X(x,x')-d_Y(y,y') \vert : (x,y), (x',y')\in R\} .
\end{equation}
\end{definition}

\begin{definition}
Let $(X,d_X)$ and $(Y,d_Y)$ be bounded spaces. The {\it Gromov-Hausdorff semi-distance} between $(X,d_X)$ and $(Y,d_Y)$ is defined as
\begin{equation}
d_{GH}(X,Y)=\inf\nolimits_R \textrm{dis}\, R ,
\end{equation}
where the infimum is taken over all correspondences $R\subset X\times Y$.
\end{definition}

The Gromov-Hausdorff semi-distance behaves well for BLMS as we have

\begin{theorem}\label{vmw}
The Gromov-Hausdorff semi-distance between bounded Lorentzian-metric spaces that contain $i^0$ has the following properties:
\begin{itemize}
\item[(a)] The Gromov-Hausdorff semi-distance is non-negative. Further, the
 spaces $(X,d_X)$ and $(Y,d_Y)$ are isometric and homeomorphic, that is, there exists an isometry (which then is also a homeomorphism) $f\colon X\to Y$, if and only if $d_{GH}(X,Y)=0$.
\item[(b)] $d_{GH}(X,Y)=d_{GH}(Y,X)$
\item[(c)] $d_{GH}(X,Z)\le d_{GH}(X,Y)+d_{GH}(Y,Z)$
\end{itemize}
The same result holds with ``contain" replaced by ``do not contain''.
\end{theorem}

To obtain a true distance, one must pass to equivalence classes. Through the study of causets and $\epsilon$-nets a precompactness theorem can be formulated \cite{minguzzi22}.

The convergence of non-compact spaces is defined by analogy with pointed metric spaces \cite{burago01}.
A key distinction emerges here: while metric geometry employs a single basepoint, the Lorentzian setting requires a countable sequence.

This need for a sequence, rather than a point, is essential to prevent the loss of information near the boundary. A simple analogy is covering the strip between two parallel spacelike lines in Minkowski 1+1 spacetime with causal diamonds; this task requires infinitely many diamonds, reflecting the need for a sequence to adequately capture the spacetime's structure.

%A key distinction emerges here: whereas metric geometry requires only one basepoint, the Lorentzian setting requires a countable sequence of points.

%This requirement for a sequence, rather than a single point, is essential to prevent the loss of information near the boundary.

%For non-compact LMS we need the analog of {\bf pointed metric space}.

Ultimately the analog of {\em pointed metric space} is just a countably generating LMS for which there is a privileged sequence.

\begin{definition}
    We say that $(X,d,\{p^k\}_{k\in\mathbb{N}})$ is a \emph{sequenced} Lorentzian metric space if $(X,d)$ is a Lorentzian metric space and $\{p^i\}_{i\in \mathbb{N}}$ is a  generating set of $X$.
    %We will refer to $\{p^k\}_{k\in\mathbb{N}}$ as a generating sequence.
\end{definition}

\begin{definition}
    We say that two sequenced Lorentzian metric spaces
    \[
    (X,d,\{p^k\}_{k\in \mathbb{N}}) \ \ \textrm{and} \ \
      (X',d',\{p'{}^k\}_{k\in \mathbb{N}})
     \]
       are \emph{isomorphic} if there exists a bijective distance preserving map $\phi: X \to X'$ such that $p'{}^n=\phi(p^n)$ for every $n\in \mathbb{N}$.
       %We also say that such a map $\phi$ is an \emph{isomorphism}.
\end{definition}

Through a somewhat technical notion of $(m,\delta)$-quasi-correspondence \cite[Sec.\ 6.1]{minguzzi24b} (omitted) it is possible to define the GH-convergence for sequenced LMS

\begin{definition}[GH-convergence]
We say that the sequence of sequenced Lorentzian metric spaces $({X}_n,d_n,\{p_n^m\}_{m\in \mathbb{N}})$ GH-converges to the sequenced Lorentzian metric space $({X},d,\{p^m\}_{m\in \mathbb{N}})$ if, for every $m$ and $\delta>0$, there exists $n_0$, such that for every $n \ge n_0$ there exists an $(m,\delta)$ $X_n-X$ quasi-correspondence.
\end{definition}
This definition is compatible with that for BLMS.

\begin{proposition}[Limit uniqueness]
Suppose that a sequence of sequenced Lorentzian metric spaces $({X}_n,d_n,\{p_n^k\}_{k\in\mathbb{N}})$  $GH$-converges  to both the sequenced Lorentzian metric spaces  $({X},d,\{p^k\}_{k\in\mathbb{N}})$ and $({X}',d',\{p'^k\}_{k\in\mathbb{N}})$, then ${X}$ and ${X}'$ are isomorphic.
\end{proposition}

Some results \cite{minguzzi24b} explore the dependence of GH-limits on the choice of generating sequences.

The main result of our study,   which motivated our very definition of (B)LMS, is

\begin{theorem}[valid for both BLMS or sequenced LMS]
    Let $X_n$ be a sequence of Lorentzian (pre)length spaces GH-convergent to a Lorentzian metric space $X$. Then $X$ is a Lorentzian (pre)length space.
\end{theorem}

Results on the stability under GH-convergence of bounds on sectional curvature were also obtained \cite{minguzzi22}.

\section{The canonical quasi-uniform structure}
\label{sec:quasi}

For a Lorentzian metric space $X$ we define a map (Kuratowski embedding)
\[
\kappa: X\to C(X)\times C(X),
\]
\[
x\mapsto (d_x,d^x).
\]
 $C(X)$ is regarded as a Frechet space with the locally convex topology generated by the semi-norms indexed by compact subsets of $X$
\[
\Vert f\Vert_{K}=\sup_{x\in K} \vert f(x)\vert.
\]
The following result was first proved for BLMS in order to show the second-countability of their topology. However, it holds also for LMS

\begin{lemma}
    Let $(X,d)$ be a Lorentzian metric space without chronological boundary. Then $\kappa: X\to C(X)\times C(X)$ is a homeomorphism into the image.
\end{lemma}

An interesting result, which we are going to discuss, establishes  that every Lorentzian metric space carries a natural quasi-uniform structure, and every sequenced LMS is quasi-metrizable.

A \emph{quasi-uniformity} $\mathcal{Q}$ on a set $X$ is a filter of subsets of $X \times X$ containing the diagonal and satisfying a condition analogous to the triangle inequality. It induces a topology and a preorder.

 Let $X$ be a set. A \emph{quasi-uniformity} (or a \emph{quasi-uniform structure}) $\mathscr{Q}$ is a filter on the set $X\times X$ such that
    \begin{enumerate}
        \item Every element $U\in \mathscr{Q}$ contains the diagonal,
        \[
        \Delta=\{(x,x): \ x\in X\}\subset U .
        \]
        \item  For every $U\in \mathscr{Q}$ there is $V\in \mathscr{Q}$ such that $V\circ V\subset U$.
    \end{enumerate}
    If, in addition to that, $\mathscr{Q}$ satisfies
    \begin{enumerate}
        \item[3.] for every $U\in \mathscr{Q}$, $U^{-1}\in \mathscr{Q}$, where  $U^{-1}:=\{(x,y): \ (y,x)\in U\}$.
    \end{enumerate}
    we say that $\mathscr{Q}$ is an \emph{uniformity} (or a \emph{uniform structure}).

If $\mathscr{Q}$ is a quasi-uniformity, then
\[
    \mathscr{U}=\mathscr{Q}^*:=\left\{U \cap V^{-1} : \ U,V\in \mathscr{Q} \right\}
\]
is a uniformity,
and
\[
    G=\bigcap\mathscr{Q}:=\cap_{U\in\mathscr{Q}}\,U
\]
is a preorder (reflexive and transitive relation).
The uniformity is Hausdorff if and only if the preorder is antisymmetric  (hence an order). The uniformity induces a topology as follows: the topology $\mathcal{T}(\mathscr{U})$ is that which assign to $x\in X$ the neighborhoods system given by $\{U[x], U\in \mathscr{U}\}$ where $U[x]=\{y\in X: (x,y)\in U\}$.

So, any quasi-uniform space induces a topological preordered space
\[
\left(X,\mathcal{T}(\mathscr{Q}^*),\bigcap\mathscr{Q}\right),
\]
 which turns out to be a closed preordered space as the preorder $G=\bigcap\mathscr{Q}$ is closed in the product topology $\mathcal{T}(\mathscr{Q}^*) \times \mathcal{T}(\mathscr{Q}^*)$.

Remarkably, every LMS comes with a canonically associated quasi-uniformity.

\begin{theorem}[Canonical Quasi-Uniformity]
Let $(X, d)$ be an LMS without boundary. The initial quasi-uniformity of the family of maps
\[
d_p: (X, \mathscr{T}) \to \big(\mathbb{R}, \mathcal{Q}_{\mathbb{R}}\big) \quad \text{and} \quad d^p: (X, \mathscr{T}) \to \big(\mathbb{R}, \mathcal{Q}^{-1}_{\mathbb{R}}\big)
\]
is a quasi-uniformity whose associated topology is $\mathscr{T}$ and whose associated order is the causal relation $J$. Here, $\mathcal{Q}_{\mathbb{R}}$ is the canonical quasi-uniformity of $\mathbb{R}$ generated by the sets of the form $\{(x,y) : -t < y-x\}$, $t>0$.
\end{theorem}

More explicitly, the quasi-uniformity on $X$ is generated by sets of the form
\begin{equation}
\big\{(x,y): \quad  (d_y-d_x)(p)<a, \ \  (d^x-d^y)(p)<b \big\},
\end{equation}
for any $p\in X$ and $a,b\in (0,+\infty)$.
Actually, there is a second canonical quasi-uniformity called the {\em fine quasi-uniformity} generated by sets of the form
\begin{equation}
\Big\{(x,y): \quad  \sup_K (d_y-d_x)<a, \ \  \sup_K (d^x-d^y)<b \Big\},
\end{equation}
where $K\subset X$ is any compact subset and $a,b\in (0,+\infty)$. It still has the property that its induced topology and order are those of the LMS.

For a sequenced LMS $(X, d, (p^k))$, we can go further and define an explicit \emph{quasi-metric}.

A {\em quasi-metric}   on a set $X$ is
a function $p:X\times X \to [0,+\infty)$ such that for $x,y,z\in X$
\begin{itemize}
\item[(i)] $p(x,x)= 0$,
\item[(ii)] $p(x,z)\le p(x,y)+p(y,z)$,
\item[(iii)] $p(x,y)=p(y,x)=0 \Rightarrow x=y$
\end{itemize}

%The quasi-pseudo-metric is called {\em pseudo-metric} if
%$p(x,y)=p(y,x)$.  It is called a {\em (Albert's) quasi-metric} if $p(x,y)=p(y,x)=0 \Rightarrow x=y$.
If $p$ is a quasi-metric then $q$, defined
by
\[q(x,y)=p(y,x),\] is a quasi-metric called {\em conjugate} of
$p$ often denoted $p^{-1}$.

A quasi-metric induces a Hausdorff quasi-uniformity via the filter determined by the countable base $W_n=\{(x,y)\in X\times X: p(x,y)<1/n\}$ and directly an order $G$ and metric $h$ (hence a topology) via
\begin{align*}
G&=\bigcap
\mathcal{U}=\{(x,y): p(x,y)=0\},\\
h&=p+p^{-1}.
\end{align*}

The canonical quasi-metric of a sequenced LMS (recall that every countably generated LMS can be sequenced) is
\begin{equation}
p(x, y) = \sum_{m=1}^{\infty} 2^{-m} \frac{ \|d_y - d_x\|^+_{X^m} }{1 + \|d_y - d_x\|^+_{X^m} } + \sum_{m=1}^{\infty} 2^{-m} \frac{ \|d^x - d^y\|^+_{X^m} }{1 + \|d^x - d^y\|^+_{X^m} },
\end{equation}
where $X^m = \overline{I_R(p^1, \ldots, p^m)}$ is the closure of the $I_R(:=I\cup \Delta)$-chronologically convex hull of $\{p^1, \ldots, p^m\}$ and $\|f\|^+_K = \sup_K \max(f, 0)$. This quasi-metric induces the fine canonical quasi-uniformity. Its symmetrization $h = p + p^{-1}$ is a metric that is Lipschitz-equivalent to the metric obtained from the Kuratowski embedding $x \mapsto (d_x, d^x)$ .

This structure provides a  unifying framework: the topology is given by $p+p^{-1}$, the order by $p^{-1}(0)=J$, and the convergence of curves can be characterized using the quasi-uniformity. In fact, by using this structure we could recover the results on limit curves previously mentioned.
%
%\section{Conclusion}
%\label{sec:conclusion}
%
%We have introduced a comprehensive theory of Lorentzian metric spaces founded solely on the Lorentzian distance function $d$. Our definition is minimal, natural, and captures the essential properties of globally hyperbolic spacetimes. The theory includes:
%\begin{itemize}
%    \item A unique, well-behaved (Polish) topology derived from $d$.
%    \item A canonical causal relation $J$ derived from $d$.
%    \item A robust notion of Gromov-Hausdorff convergence for non-compact spaces.
%    \item Stability of the (pre)length space property under GH-limits.
%    \item A canonical quasi-uniform and quasi-metric structure.
%\end{itemize}
%
%This framework provides a solid foundation for studying the convergence of spacetimes, the geometry of causal sets, and the development of a synthetic Lorentzian geometry in analogy to the metric geometry of Alexandrov and Ricci limit spaces.
%
%
%

%S.G. work is partially supported by ...

%\noindent {\bf Data Availability Statement}. All data generated or analysed during this study are included
%in this published article.

\section{Existence of Cauchy time functions}

In this section, we present new material that addresses the existence of Cauchy time functions on Lorentzian metric spaces (LMS). While we provide all proofs, we assume the reader has some familiarity with the theory of LMS.

\begin{definition}
An isocausal curve $\gamma: I \to X$, $I=(a,b)$ or $I=[0,b)$, has {\em future endpoint} $p$ if $\lim_{t\to b} \gamma(t)=p$ ($b= +\infty$ is allowed). We have an analogous definition of past endpoint.
\end{definition}

\begin{proposition} \label{wxi}
Let $(X,d)$ be a Lorentzian metric space, $X=I(X)$. Suppose for the isocausal curve $\gamma: I \to X$, $I=(a,b)$ or $I=[0,b)$, there is a sequence $b_k\to b$ such that $\lim_k \gamma(b_k)=p$ for some $p\in X$. Then $\gamma$ has future endpoint $p$. An analogous version holds in the past case.
\end{proposition}

\begin{proof}
Let $C$ be a compact neighborhood of $p$. Since $\gamma(b_k)\to p$, either there is $T>0$ such that for $s>T$, $\gamma(s)\in C$, or the curve enters and escapes indefinitely $C$. In the former case, if there is no convergence then we can find an open set $O\ni p$ and a sequence $e_k\to b$,  such that $\gamma(e_k)\in C\backslash O$,   $\gamma(e_k)\to q\in C\backslash O$, hence $q\ne p$.

The latter case means that we can find $c_k,d_k\to b$, $c_k<d_k<c_{k+1}$, such that $\gamma(c_k)\in C$ and $\gamma(d_k)\notin C$. By the continuity of the curve we can find $e_k\to b$, $c_k< e_k <d_k$, such that $\gamma(e_k)\in \p C$,  $\gamma(e_k)\to q \ne p$.

In both cases, for every $k$, $\gamma(b_k)\le \gamma(e_{s(k)})\le \gamma(b_{r(k)})$ for suitable monotone functions $s(k)$ and $r(k)$, thus taking the limit $p\le q\le p$, which implies $p=q$, a contradiction.
\end{proof}

\begin{definition}
An isocausal curve $\gamma: I \to X$, $I$ connected interval of the real line, is said to be {\em future inextendible} if it has no future endpoint. We have analogous definitions of {\em past inextendible} and {\em inextendible}.
\end{definition}

\begin{definition}
A future inextendible isocausal curve $\gamma: [0,b)\to X$, is {\em future total  imprisoned} in a compact set $C$ if there is $B\in [0,b)$ such that $\gamma(t)\in C$ for every $t>B$. We have an analogous definition for  {\em past total  imprisoned}.
\end{definition}

\begin{definition}
A Lorentzian metric space is {\em non-total imprisoning} if there is no future inextendible isocausal curve future total imprisoned in a compact set and similarly in the past case.
\end{definition}

\begin{definition}
A future inextendible isocausal curve $\gamma: [0,b)\to X$, is {\em future partial  imprisoned} in a compact set $C$ if there are $b_k\to b$, such that $\gamma(b_k)\in C$. We have an analogous definition for  {\em past partial  imprisoned}.
\end{definition}

\begin{definition}
A Lorentzian metric space is {\em non-partial imprisoning} if there is no future inextendible isocausal curve future partial imprisoned in a compact set and similarly in the past case.
\end{definition}

Clearly, a curve which is future total imprisoned in a compact set is also future partial imprisoned in it, thus the non-partial imprisonment property implies the non-total imprisonment property.

\begin{proposition}
Every Lorentzian metric space $(X,d)$, $X=I(X)$, is non-partial imprisoning.
\end{proposition}

\begin{proof}
Suppose that there is a future inextendible isocausal curve $\gamma:[0,b) \to X$ future imprisoned in a compact set $C$. Let $b_k\to b$, $\gamma(b_k)\in C$, then we can pass to a subsequence, denoted in the same way, such that $\gamma(b_k)\to p\in C$. By Prop.\ \ref{wxi} $p$ is a future endpoint for $\gamma$, a contradiction with its future inextendibility.
\end{proof}

Note that the term inextendibility refers to the absence of  endpoints. The fact that in presence of endpoints the curve can be extended requires further conditions.

\begin{proposition}
Let $(X,d)$ be a countably-generated Lorentzian prelength space, $X=I(X)$. Every isocausal curve $\gamma: I \to X$, $I=(a,b)$ or $I=[0,b)$, $b<\infty$,  with future endpoint is the restriction of a future inextendible isocausal curve. The case $b=+\infty$ requires a preliminary reparametrization of the curve to accomplish $b<+\infty$. An analogous version holds for the past inextendible and inextendible results.
\end{proposition}

\begin{proof}
We consider the case $I=[0,b)$, the other being a consequence.
Let $\tau: X\to \mathbb{R}$ be a time function, which exists by the countably generated condition. The function $\tau\circ \gamma$ is continuous and monotone. Note that $\gamma(b_k)\le \gamma(b_{k+1})$ thus $\tau(\gamma(b_k))\le \tau(\gamma(b_{k+1}))$, where the sequence converges to $\tau(p)$ thus $\tau \circ \gamma$ is bounded by $\tau(p)<+\infty$. The increasing continuous function $f=\tau\circ \gamma:[0,b)\to (\tau(\gamma(0)), \tau(p))$ can be extended to an increasing continuous function $\tilde f:[0,+\infty)\to (\tau(\gamma(0)), +\infty)$. Let $\tilde \gamma(t):=\gamma(f^{-1}(t))$ which is the reparametrization with time function $\tau$ of $\gamma$. It is $\tau$-uniform in the sense that $\tau(\tilde\gamma(t))=t$. This isocausal curve can be extended. Let $q \in I^+(p)$ then there is an isocausal curve $\sigma$ connecting $p$ to $q$ which we can parametrize so as to become $\tau$-uniform. Now, $\gamma$ can be extended adjoining it to $\sigma$.

Next, we conclude with a Zorn type argument. Consider the family of all $\tau$-uniform extensions ordered by the inclusion of the images. There is a maximal chain and its union provides a $\tau$-uniform curve $\eta(t)$ which cannot be extended, otherwise the chain would not have been maximal. The curve $\eta(\tilde f(s))$ provides a future inextendible extension to the original curve.
\end{proof}

\begin{definition}
A {\em Cauchy time function} is a time function $\tau$ such that, once composed with a future inextendible isocausal curve $\gamma: [0, \infty) \to X$ gives
a proper  function $\tau \circ \gamma: [0,\infty) \to \mathbb{R}$, and once composed with a past inextendible isocausal curve $\gamma: (-\infty, 0] \to X$ gives
a proper  function $\tau \circ \gamma: (-\infty, 0] \to \mathbb{R}$.
\end{definition}

Equivalently,  in the former future inextendible case $\tau(\gamma(t)) \to +\infty$ for $t\to +\infty$, and in the latter past inextendible case $\tau(\gamma(t))\to -\infty$ for $t \to -\infty$. In particular, if $\gamma$ is inextendible, the image of $\tau \circ \gamma$ is $\mathbb{R}$.

%In other words, the image of $\tau \circ \gamma$ is $\mathbb{R}$.

The following proof is interesting already in the smooth case. It avoids use of Geroch's volume time functions \cite{geroch70,hawking73}, hence of any measure on $X$, and relies solely on the continuity of $d$, the compactness of diamonds, and the distinguishing property. In other words, the defining properties of a LMS enter neatly in the proof.

\begin{theorem}
Every countably generated Lorentzian metric  space $(X,d)$,  admits a Cauchy time function $\tau$. Moreover, $\tau$ can be chosen {\em rushing}, namely such that
\[
\tau(x)+d(x,y) \le \tau(y)
\]
 whenever $x\le y$.
\end{theorem}

We recall that for a countably generated Lorentzian metric space $X=I(X)$.

\begin{proof}
We know that there is a dense countable family $\mathcal{S}=\{z_k\}$ that distinguishes points, in the sense that, if $x\ne y$, there is some $z_k$ such that $d(z_k, x) \ne d(z_k,y)$ or $d(x, z_k)\ne d(y,z_k)$. In particular, by the countably generated property, the density and the openness of the chronological relation, every point admits some element of $\mathcal{S}$ in its chronological past and some element  of $\mathcal{S}$ in its chronological future.
Let us introduce the functions
\[
f(x):= \sum_{k=1} \frac{1}{2^k} \frac{d(z_k, x)}{1+d(z_k, x)}, \qquad g(x):= \sum_{k=1} \frac{1}{2^k} \frac{d(x,z_k)}{1+d(x,z_k)} ,
\]
then $f: X \to (0,1]$ is continuous and isotone and $g: X \to (0,1]$ is continuous and anti-isotone. Note that none of them takes the value zero by the property recalled above on the existence of some element $z_k$ in the past/future of $x$. The function $f/g$ is thus well-defined, continuous and isotone. Moreover, it is a time function, because if $x<y$, there must be some point $z_k$ that distinguishes them. If  $d(z_k, x) \ne d(z_k,y)$ then necessarily $d(z_k,x)<d(z_k, y)$, which implies that $f(x)<f(y)$, $g(x)\ge g(y)$ and so $\frac{f}{g}(x)<\frac{f}{g}(y)$. Similarly, if  $d(x, z_k)\ne d(y,z_k)$ then necessarily $d(x,z_k)> d(y,z_k)$, which implies that $g(x)>g(y)$, $f(x)\le f(y)$ and so $\frac{f}{g}(x)<\frac{f}{g}(y)$.
Let us prove that $\tau=\log (f/g)$ is a Cauchy time function. This would follow by proving that, for any inextendible isocausal curve $\gamma: \mathbb{R}\to X$,  $f(\gamma(s))\to 0$ for $s\to -\infty$ and $g(\gamma(s))\to 0$ for $s\to +\infty$. We prove the latter property, the former being analogous.  We just need to show that for every $k$, the future inextendible curve is bound to definitively  escape $I^-(z_k)$, so that for sufficiently large $s$, $d(\gamma(s),z_k)=0$ and the $k$-th contribution to $g$ vanishes. This follows from the compactness of $J^+(\gamma(0))\cap J^-(z_k)$ and the fact that non-partial imprisonment holds so that $\gamma(s) \in J^+(\gamma(0))$ does not belong to $J^-(z_k)$ for sufficiently large $s$.

Observe that the sum of a Cauchy time function and a continuous rushing function provides a rushing Cauchy time function, thus we need only to prove existence of a continuous rushing  function. Rushing time functions (hence continuous) exist by the main result of \cite{minguzzi25c} which,  as observed there, applies to LMS (see Ex.\ 3.13 and  Thm 5.8 of that reference).
\end{proof}

\section*{Acknowledgments }I  take the opportunity to thank  A. Bykov and S. Suhr for the fruitful collaboration. This review was presented at the  BIRS-IMAG meeting in Granada.
I am  grateful to L. García-Heveling for posing the question of whether Cauchy time functions exist on Lorentzian metric spaces, which motivated the original part of this paper. I also thank the organizers of the  meeting—A. Shadi Tahvildar-Zadeh, A. Burtscher, J. L. Flores, and L. Mehidi—for hosting a stimulating and productive conference.

Finally, I thank M. Braun and V. Dardel for a question on a sloppy statement about the continuity of $\tilde d$ in the opening paragraphs of \cite[Sec.\ 1.3]{minguzzi22}  and  A. Bykov and S. Suhr for help in clarifying the matter (see discussion after Thm.\ \ref{tjg}).

This study was funded by the European Union - NextGenerationEU, in the framework of the PRIN Project (title) {\em Contemporary perspectives on geometry and gravity} (code 2022JJ8KER - CUP B53D23009340006). The views and opinions expressed in this article are solely those of the authors and do not necessarily reflect those of the European Union, nor
can the European Union be held responsible for them.

%\bibliography{../../bibliografie/simultaneity,../../bibliografie/libri,../../bibliografie/miei,../../bibliografie/mieiPrep}
%\bibliographystyle{plain}
%

\end{document}